\documentclass[conference]{IEEEtran}
\ifCLASSINFOpdf
\else
\fi
\usepackage[utf8]{inputenc}
\usepackage{amsmath}
\usepackage{amsfonts}
\usepackage{amssymb}
\usepackage{graphicx}
\usepackage{float}
\usepackage{mathrsfs}
\usepackage{algorithm}
\usepackage[noend]{algpseudocode}
\usepackage{graphicx}
\usepackage[caption=false]{subfig}
\usepackage{xcolor}
\usepackage{colortbl}
\usepackage{pbox}
\usepackage{cite}
\usepackage{lipsum}
\usepackage{amsthm}
\usepackage[T1]{fontenc}
\usepackage[utf8]{inputenc}
\usepackage{subfig}
\newtheorem{corollary}{Corollary}
\newtheorem{lemma}{Lemma}

\usepackage{csquotes}
\IEEEoverridecommandlockouts

\usepackage{subfig}
\def\href#1#2{#1 #2}
\ifCLASSINFOpdf

\else

\fi

\begin{document}
\bstctlcite{IEEEexample:BSTcontrol}

\title{Wireless Communications and Control for Swarms  of Cellular-Connected UAVs\vspace{-0.50cm}}

\author{\IEEEauthorblockN{Tengchan Zeng$^1$, Mohammad Mozaffari$^2$, Omid Semiari$^3$, Walid Saad$^1$, Mehdi Bennis$^4$, and Merouane Debbah$^5$}\vspace{-0.40cm}\\
	\IEEEauthorblockA{
		\small $^1$ Wireless@VT, Department of Electrical and Computer Engineering, Virginia Tech, Blacksburg, VA, USA,\\ Emails:\{tengchan, walids\}@vt.edu.\\
		$^2$ Ericsson, Santa Clara, CA, USA, Email:  mohammad.mozaffari@ericsson.com.\\
		$^3$ Department of Electrical and Computer Engineering, Georgia Southern University, Statesboro, GA, USA,\\ Email: osemiari@georgiasouthern.edu.\\
		$^4$ Centre for Wireless Communications, University of Oulu, Oulu, Finland, Email:
		mehdi.bennis@oulu.fi. \\
		$^5$ Mathematical and Algorithmic Sciences Lab, Huawei France R\&D, Paris, France, Email: merouane.debbah@huawei.com.\vspace{-0.80cm}
		\thanks{This research was supported by the U.S. National Science Foundation under Grant CNS-1739642.}
}}
\maketitle

\begin{abstract}
By using wireless connectivity through cellular base stations (BSs), swarms of unmanned aerial vehicles (UAVs) can provide a plethora of services ranging from delivery of goods to surveillance. In particular, UAVs in a swarm can utilize wireless communications to collect information, like velocity and heading angle, from surrounding UAVs for coordinating their operations and maintaining target speed and intra-UAV distance. 
However, due to the uncertainty of the wireless channel, wireless communications among UAVs will experience a transmission delay which can impair the swarm's ability to stabilize system operation. 
In this paper, the problem of joint communication and control is studied for a swarm of three cellular-connected UAVs positioned in a triangle formation. 
In particular, a novel approach is proposed for optimizing the swarm's operation while jointly considering the delay of the wireless network and the stability of the control system. 
Based on this approach, the maximum allowable delay required to prevent the instability of the swarm is determined. 
Moreover, by using stochastic geometry, the reliability of the wireless network is derived as the probability of meeting the stability requirement of the control system.   
The simulation results validate the effectiveness of the proposed joint strategy, and help obtain insightful design guidelines on how to form a stable swarm of UAVs.\vspace{-0.20cm}

\end{abstract}

\section{Introduction}
The deployment of unmanned aerial vehicles (UAVs), popularly known as drones, is rapidly increasing and will lead to the introduction of numerous application services ranging from delivery of goods to surveillance and smart city monitoring \cite{mozaffari2018tutorial}. 
In particular, driven by the ever-decreasing cost of manufacture components and the emergence of new services, the use of UAV swarms is rapidly gaining popularity \cite{bamburry2015drones,mozaffari2018beyond,mozaffari2018communications,burkle2011towards,mozaffari2016efficient,esrafilian2018simultaneous}.   

In addition, by using swarms of UAVs, one can complete more sophisticated missions in an uncertain and possibly hostile environment.   
For example, swarms of drones have been used for Amazon's prime air drone delivery services and emergency medicine delivery services \cite{bamburry2015drones}. 
Also, in \cite{mozaffari2018communications}, groups of UAVs have been used to create a reconfigurable antenna array in the sky so as to provide wireless service to ground users. 
Moreover, swarms of micro drones are being actively investigated by DAPRA's VisiBuilding program to complete reconnaissance missions inside buildings \cite{burkle2011towards}. 
Furthermore, the authors in \cite{mozaffari2016efficient} and \cite{esrafilian2018simultaneous} proposed to deploy multiple UAVs, which function as wireless base stations (BSs) or relays, so as to maximize the wireless coverage.

In particular, to complete their assigned missions, the UAVs in the swarm will have to communicate with BSs via cellular links for various purposes such as sending collected surveillance and monitoring information back to the BSs. 
Moreover, to guarantee a safe operation and avoid collisions between UAVs within the swarm, the UAVs will use the information received from the intra-swarm wireless network as an input of the control system. 
That is, each UAV can first use intra-swarm communications to obtain information of other UAVs in the swarm, such as their velocity and heading angle.
Then, the control system of each UAV will use sensor data and the information received from the wireless links to coordinate UAV's movements. 
However, due to the uncertainty of the wireless channel and the presence of wireless interference, the received information from wireless links will inevitably suffer from transmission delay, and the delayed information can impair the ability of the control system to coordinate the UAVs' movements \cite{gupta2016survey}. 
As a result, when designing a swarm of UAVs, we need to jointly consider the control system and wireless network to guarantee a stable formation. 

The prior art working on wirelessly connected swarms of UAVs can be grouped into two categories. 
In particular, the first category focuses on the intra-swarm communication network design \cite{zafar2017reliable} and \cite{li2008robot}. 
For example, the authors in \cite{zafar2017reliable} proposed a multicluster flying ad-hoc network to reduce the power consumption while maintaining an acceptable level of communication latency for UAV swarms.
Furthermore, a wireless mesh network is proposed in \cite{li2008robot} to improve the connectivity of a swarm of UAVs and build a pervasive networking environment.
However, prior works, such as \cite{zafar2017reliable} and \cite{li2008robot}, ignore the impact of wireless system on the stability of the UAV and solely focus on the communication system design. 
The second category focuses on coordination and control for effective task planning for UAVs \cite{bekmezci2014connected} and \cite{ben2008distributed}.
For example, in \cite{bekmezci2014connected}, a heuristic multi-UAV task planning algorithm is proposed to enable a swarm of cellular connected UAVs to visit all target points in a minimum time. 
In addition, the authors in \cite{ben2008distributed} proposed a behavioral flocking algorithm for distributed flight coordination of multiple UAVs. 
Note that, such control-centric works, like   \cite{bekmezci2014connected} and \cite{ben2008distributed}, assume a fixed wireless performance or just ignore the transmission delay generated by communication links when designing the control system. 
Such an assumption is certainly not practical for UAV swarms that use a cellular network due to the uncertainty of wireless channels and interference generated by coexisting wireless links.  

The main contribution of this paper is a novel approach to jointly design the control and communication system for a cellular-connected swarm of UAVs.
In particular, we first analyze the stability of the control system which can guarantee a stable triangle formation for a swarm of three UAVs. 
Then, we determine the maximum transmission delay that the considered swarm can tolerate without jeopardizing its control system's stability. 
This threshold can, in turn, be used to identify the reliability requirement for the wireless communication system. 
In particular, we use stochastic geometry to mathematically characterize the reliability of the wireless network. 
Simulation results validate the effectiveness of the proposed integrated communication and control strategy, and help obtain new design guidelines on how to create a stable formation for a swarm of UAVs.
For example, our results provide clear guideline on how to choose the target spacing for the swarm so as to guarantee a target reliability performance for the wireless network. 

The rest of this paper is organized as follows. Section \uppercase\expandafter{\romannumeral2} presents the system model. 
In Section \uppercase\expandafter{\romannumeral3}, we perform a stability analysis for the control system for the swarm of UAVs and derive the mathematical expression for the reliability of the wireless network by using stochastic geometry. Section \uppercase\expandafter{\romannumeral4} provides the simulation results and conclusions are drawn in Section \uppercase\expandafter{\romannumeral5}.\vspace{-0.3cm}
\begin{figure}[!t]
	\centering
	\includegraphics[width=2.8in,height=2.2in]{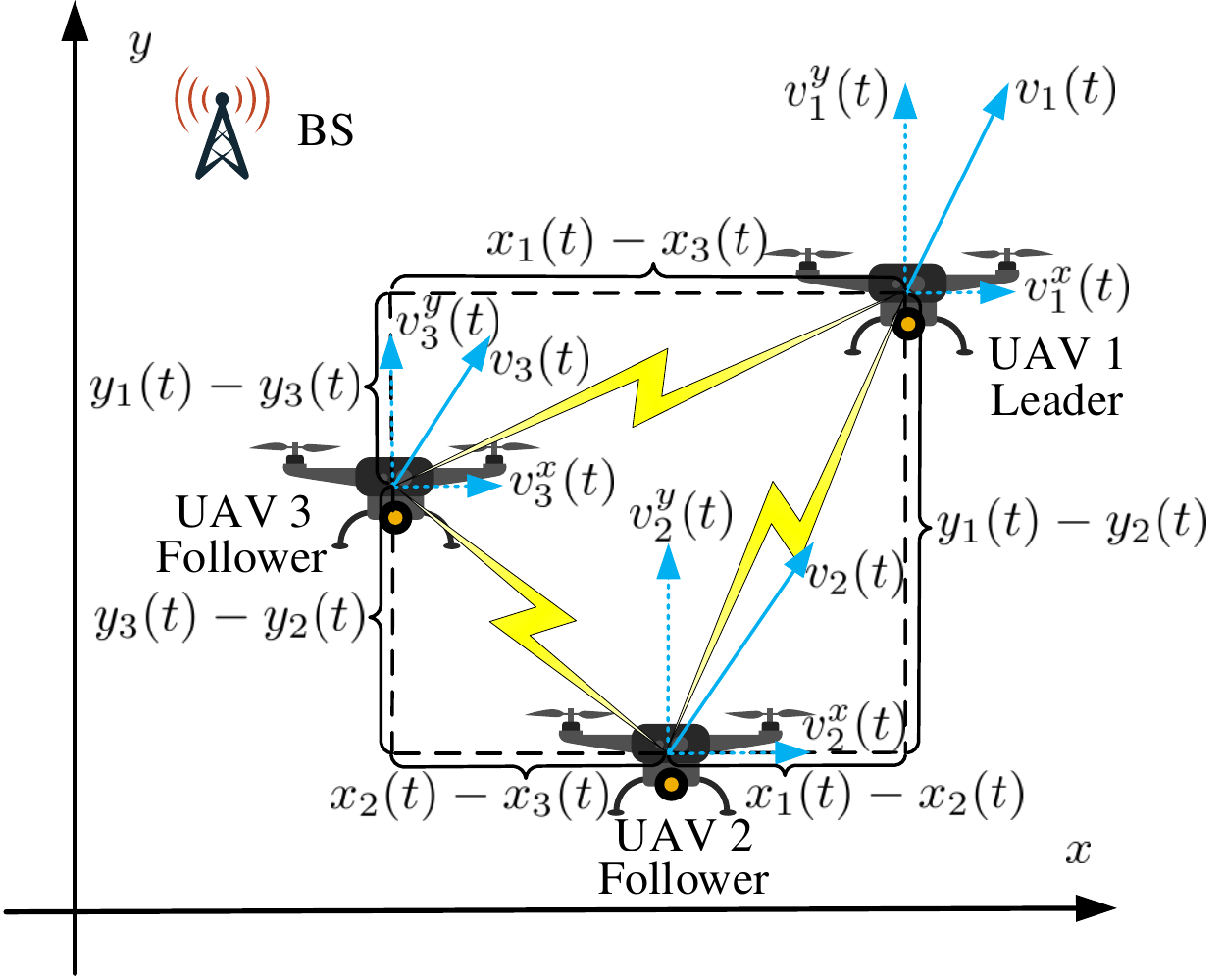}
	\vspace{-0.3cm}
	\DeclareGraphicsExtensions.
	\caption{A swarm of three UAVs where each UAV can communicate with the BS. UAV 1 is the leader, and UAVs 2 and 3 are followers.}
	\label{systemmodel}
	\vspace{-0.6cm}
\end{figure}

\section{System Model}
Consider a swarm of three UAVs flying at the same altitude. 
In this swarm, two UAVs are following a third, leading UAV to form and maintain a triangle formation, as shown in Fig. \ref{systemmodel}.
In this model, we assume that UAV $1$ is the leading UAV that always flies with a target velocity and heading direction, while UAVs $2$ and $3$ are followers.
Note that, when following the leading UAV, UAVs $2$ and $3$ will also need to keep a certain target distance with each other and with the leader. 
For each UAV, an embedded radar sensor can sense the distance to the nearby UAVs. 
Moreover, each following UAV can communicate with the two other UAVs in the swarm, via wireless cellular links, to obtain information of the velocity and heading angle. 
In addition, each UAV will communicate with the BSs through cellular links to report its movement and location or to complete tasks, like sending the collected surveillance information back to the BSs.

\subsection{Control System Model}
As shown in Fig. \ref{systemmodel}, we consider a Cartesian coordinate system centered on an arbitrarily selected point, and the location of each UAV at time $t$ is denoted by $(x_{i}(t),y_{i}(t)), i \in\{1,2,3\}$.
Also, by using the coordinate system, we can decompose the velocity of each UAV into two components: one on the x-axis and the other one on the y-axis. 
Moreover, we assume that the control laws of each following UAV over the x-axis and y-axis depend on the difference between the actual distance to the UAVs in the swarm and the target distance. 
For example, if the actual distance between a following UAV and other  UAVs in the swarm exceeds the corresponding target spacing, the following UAV needs to accelerate in order to reduce the spacing and reach the target distance. 
To determine the control law on each axis, we first take the component on the x-axis as an example. 
In particular, we define the x-axis distance difference by using the following spacing errors: 
\begin{align}
\label{spacingerror1}
\delta_{1,2}^{x}(t)\!\! =\!\! x_{1}(t)\!\! -\!\! x_{2}(t)\!\! -\!\!\bar{x}_{1,2}, 
\delta_{1,3}^{x}(t)\!\! =\!\! x_{1}(t)\! \!-\!\! x_{3}(t) \!\!-\!\!\bar{x}_{1,3},\\
\label{spacingerror2}
\delta_{2,3}^{x}(t) \!\!=\!\! x_{2}(t) \!\!-\!\! x_{3}(t) \!\!-\!\!\bar{x}_{2,3},
\delta_{3,2}^{x}(t) \!\!= \!\!x_{3}(t)\!\! -\!\! x_{2}(t) \!\!-\!\!\bar{x}_{3,2},
\end{align}
where $\bar{x}_{1,2}, \bar{x}_{1,3}, \bar{x}_{2,3},$ and $\bar{x}_{3,2}$ are the target x-axis spacing between the corresponding UAVs.
Note that $x_{i}(t)-x_{j}(t)$ is usually considered as the x-axis  headway distance between UAVs $i$ and $j$ with $i,j\in \{1,2,3\}$, at time $t$.
Also, we can define the x-axis velocity errors for following UAVs $2$ and $3$ as 
\begin{align}
\label{velocityerror1}
&z^{x}_{2}(t) = v_{2}^{x}(t) - \bar{v}_{x}, z^{x}_{3}(t) = v^{x}_{3}(t)- \bar{v}_{x},
\end{align}
where $\bar{v}_{x}$ is the x-axis component of the target operating velocity for the swarm of UAVs. Note that the spacing errors and velocity errors on the y-axis can be derived as done in (\ref{spacingerror1})--(\ref{velocityerror1}) and are omitted due to space limitations.

Similar to the dynamical system model introduced in \cite{bando1995dynamical}, the acceleration and deceleration of each UAV will depend on the spacing errors and velocity errors. 
In particular, the x-axis control law for each following UAV can be given by: 
\begin{align}
u^{x}_{i}(t)\!=& a_{i}\delta_{1,i}^{x}\!+\!b_{i}[v^{x}_{1}(t\!-\!\tau_{1,i}(t))\!-\!v_{i}(t)]\!+\!\hat{a}_{i}\delta_{j,i}^{x}\!+ \nonumber \\ &\!\hat{b}_{i}[v_{j}(t\!-\!\tau_{j,i}(t))\!-\!v_{i}(t)], i \!\neq\! j, i,j \!\in\! \{2,3\},
\end{align}
where $a_{i}$, $b_{i}$, $\hat{a}_{i}$, and $\hat{b}_{i}$ are the associated gains for each corresponding term, while $\tau_{j,i}$ captures the delay for the wireless link from UAV $j$ to UAV $i$.
Note that the associated gains essentially capture the sensibility of the control system to changes in distance and velocity. 
Also, since the leading UAV always flies with the target velocity and heading angle, then the solution to $v^{x}_{1}(t-\tau_{1,i}(t))=\bar{v}_{x}$ always exists.
Similarly, we can also derive the control law over the y-axis for each following UAV.
Therefore, based on the x-axis and y-axis control laws for the velocity components, we can determine how the velocity and heading angle of each following UAV should change. 

\subsection{Wireless Communication System}
For the wireless communication links between UAVs inside the swarm, we consider an orthogonal frequency-division multiple access (OFDMA) scheme where each communication link does not share the frequency resource with other links in the swarm. 
In this case, the wireless links in the swarm can coexist without suffering from interference from each other. 
However, the wireless links in one swarm can experience interference when other UAVs at the same altitude that are using the same frequency resource to transmit information with each other or with BSs via cellular links. 
To have a general interference model, we assume that the distribution of interfering UAVs at the same altitude with the swarm follows a 2-dimensional Poisson point process (2-D PPP) with density $\lambda$. 
Also, similar to \cite{matolak2012air}, we consider the wireless communication channels inside the swarm as independent Nakagami channels with parameter $\beta$, and we also model the wireless channels from interfering UAVs to UAVs inside the swarm as independent Rayleigh fading channels.  
Thus, the channel gain between a receiving UAV $i$ and a transmitter $j$ at time $t$ will be $g_{j,i}(t)=h_{j,i}(t)(d_{j,i}(t))^{-\alpha}$, where $h_{j,i}(t)$ captures the fading gain, $d_{j,i}(t)$ is the distance between UAVs $j$ and $i$, and $\alpha$ is the path loss exponent. 
Moreover, we can obtain the received signal at UAV $i$ as $P_{j,i}(t)=P_{t}g_{j,i}(t)$, where $P_{t}$ is the transmission power. Also, the signal-to-interference-plus-noise-ratio (SINR) can be given by $\gamma_{j,i}(t)=\frac{P_{j,i}(t)}{\sigma^2+I_{i}(t)}$, where $\sigma^2$ is the variance of the Gaussian noise, and $I_{i}(t)$ captures the interference experienced by UAV $i$.
Then, the data rate will be: $R_{j,i}(t)=\omega \log_{2}(1+\gamma_{j,i}(t))$, where $\omega$ is the bandwidth of the frequency resource. 
Whenever all packets are of equal size $S$ bits, the transmission delay of the wireless link between UAVs $j$ and $i$ can be derived as 
\begin{align}
\label{CommunicationDelay}
\tau_{j,i}(t)=\frac{S}{\omega \log_{2}(1+\gamma_{j,i}(t))}.
\end{align}
In the following section, we take into account the time-varying wireless transmission delay in (\ref{CommunicationDelay}) and analyze its effect on the stability of the control system in the swarm of UAVs.
\vspace{-0.05cm}
\section{Stability Analysis for the Swarm of UAVs}
For the swarm of UAVs, the delayed information received from the wireless links can negatively impact the control system's ability to coordinate the movements.
As a result, the target formation for the swarm of UAVs may fail to form. 
Here, we perform a stability analysis for the swarm under the influence of the transmission delay. 
In particular, we determine the transmission delay threshold which can guarantee that the following UAVs will fly at the same speed and heading angle with the leading UAV and keep the target distance to the other UAVs in the swarm. 
Based on the stability analysis, we employ stochastic geometry to mathematically characterize the \emph{reliability} of the wireless system, defined as the probability that the wireless system can meet the control system's delay requirements.  
\subsection{Stability Analysis}
To guarantee that each following UAV operates at the same speed and heading angle as the leading UAV and keeps a target distance to UAVs in the swarm, both spacing errors and velocity errors on the x-axis should converge to zero. 
To this end, we take the first-order derivative of (\ref{spacingerror1}), (\ref{spacingerror2}), and (\ref{velocityerror1}) as follows: 
\begin{align}\label{equation6}
\dot{\delta}_{1,2}^{x}(t)=&-z_{2}^{x}(t),\dot{\delta}_{1,3}^{x}(t)=-z_{3}^{x}(t),\\
\dot{\delta}_{2,3}^{x}(t)=&-\dot{\delta}_{1,2}^{x}(t)+\dot{\delta}_{1,3}^{x}(t), \dot{\delta}_{3,2}^{x}(t)=\dot{\delta}_{1,2}^{x}(t)-\dot{\delta}_{1,3}^{x}(t), \\
\dot{z}_{2}^{x}(t)=&(a_{2}+\hat{a}_{2}) \delta_{1,2}^{x}(t) + (-\hat{a}_{2})\delta_{1,3}^{x}(t)+\nonumber \\&(-b_{2}-\hat{b}_{2})(z_{2}^{x}(t))+\hat{b}_{2}z^{x}_{3}(t-\tau_{3,2}(t)), \\
\label{equation9}
\dot{z}_{3}^{x}(t)=&(a_{3}+\hat{a}_{3}) \delta_{1,3}^{x}(t)+ (-\hat{a}_{3})\delta_{1,2}^{x}(t)+\nonumber\\ &(-b_{3}-\hat{b}_{3})(z_{3}^{x}(t))+\hat{b}_{3}z^{x}_{2}(t-\tau_{2,3}(t)),  
\end{align}
where $\dot{\delta}_{1,2}^{x}(t), \dot{\delta}_{1,3}^{x}(t), \dot{\delta}_{2,3}^{x}(t), \dot{\delta}_{3,2}^{x}(t), \dot{z}_{2}^{x}(t)$, and $\dot{z}_{3}^{x}(t)$ are variables differentiated with respect to time $t$.
Additionally, since the channel gains of the wireless links between UAVs $2$ and $3$ and between UAVs $3$ and $2$ follow the same distribution, we assume $\tau_{2,3}(t)=\tau_{3,2}(t)=\triangle \tau(t)$.
After collecting the spacing and velocity errors for all following UAVs, we can find the augmented error state vector at the x-axis  $\boldsymbol{e}^{x}(t)=[\delta_{1,2}^{x}(t),\delta_{1,3}^{x}(t),z_{2}^{x}(t),z_{3}^{x}(t)]^{T}$ and obtain 
\begin{align}
\label{controlSystem}
\dot{\boldsymbol{e}}^{x}(t) = \boldsymbol{M}_{1}\boldsymbol{e}^{x}(t)+\boldsymbol{M}_{2}\boldsymbol{e}^{x}(t-\triangle \tau(t)),
\end{align}
where 
\begin{align}
\boldsymbol{M}_{1}=\begin{bmatrix}
0 & 0 & -1 & 0 \\
0 & 0 & 0 & -1 \\
a_2 + \hat{a}_{2} & -\hat{a}_{2} & -b_2-\hat{b}_{2} & 0 \\
-a_3 & (a_3+\hat{a}_3) & 0 & -b_3-\hat{b}_{3} 
\end{bmatrix},
\end{align}
and 
\begin{align}
\boldsymbol{M}_{2}=\begin{bmatrix}
0 & 0 & 0 & 0 \\
0 & 0 & 0 & 0 \\
0 & 0 & 0 & \hat{b}_2\\
0 & 0 & \hat{b}_3 & 0 
\end{bmatrix}.
\end{align}
Since the stability of the swarm of UAVs requires the x-axis spacing and velocity errors of all following UAVs to approach zero, the error vector $\boldsymbol{e}^{x}(t)=\boldsymbol{0}_{4 \times 1}$ should be at least asymptotically stable.

Guaranteeing the stability for a wireless-connected swarm will hence require a small wireless transmission delay. 
Therefore, next, as a direct result of \cite[Theorem 1]{zeng2018AV}, we can characterize the maximum transmission delay needed to support the convergence of error vector $\boldsymbol{e}^{x}(t)$ to the zero vector in following corollary.  

\begin{corollary}
	\label{theorem1}
	The convergence of the error vector $\boldsymbol{e}^{x}(t)$ in (\ref{controlSystem}) to the zero vector is asymptotically stable if the transmission delay $\triangle \tau(t)$ of the wireless links between these two following UAVs in the swarm satisfies: 
	\begin{align}
	\label{SINR11}
	&\triangle \tau(t)  \leq \tau_{x} =\nonumber \\ &\frac{1}{\lambda_{\max}(\boldsymbol{C}\boldsymbol{M}_{2}\boldsymbol{M}_{1}\boldsymbol{M}_{1}^T\boldsymbol{M}_{2}^T\boldsymbol{C}^T\!+\!\boldsymbol{C}\boldsymbol{M}_{2}\boldsymbol{M}_{2}\boldsymbol{M}_{2}^T \boldsymbol{M}_{2}^T\boldsymbol{C}^T\!+\!2k\boldsymbol{I})},
	\end{align}
	where $k\!>\!1$, $\boldsymbol{C}$ is a positive definitive matrix  meeting $\boldsymbol{C}(\boldsymbol{M}_{1}\!+\!\boldsymbol{M}_{2})+(\boldsymbol{M}_{1}\!+\!\boldsymbol{M}_{2})^{T}\boldsymbol{C}\!=\!-\boldsymbol{I}_{4 \times 4}$, and $\lambda_{\max}(\cdot)$ represents the maximum eigenvalue of the corresponding matrix.
\end{corollary}
Similar to the analysis in Corollary \ref{theorem1}, we can find the delay requirement $\tau_{y}$ which can guarantee that the convergence of error terms at $y$-axis to the zero vector is asymptotically stable. 
In this case, to guarantee that the error terms of the following UAVs on x and y axes converge to zero and the swarm of UAVs forms the target formation, the maximum allowable wireless transmission delay experienced by the receiving UAV should be $\min(\tau_{x},\tau_{y})$.  
\subsection{Reliability Analysis of the Wireless System}
To characterize the performance of the wireless system, we introduce a notion of reliability for the wireless system. 
In particular, we can use stochastic geometry to derive the mathematical expression for the reliability of the wireless network. 
Different from our work in \cite{zeng2018joint}, we consider that the distribution of the interfering UAVs follows a 2D-PPP. 
In the following lemma, we use stochastic properties from the 2D-PPP and calculate the reliability. 

\begin{lemma}
	\label{theorem2}
If the distribution of the interfering UAV follows a 2-D PPP with density $\lambda$, the reliability of the wireless link from UAV $j$ to UAV $i$, $i\!\neq\! j, i, j \!\in\! \{2,3\}$, can be defined as 
\begin{align}
Pr_{j,i}\!\! \approx\!\!& \sum_{k=1}^{\beta}(-1)^{k+1}{{\beta}\choose{k}}\exp\left(\frac{-k\eta \left(2^{\frac{S}{\omega \min(\tau_{x},{\tau_{y}})}}\!\!-\!\!1\right) d_{j,i}^{-\alpha}}{P_{t}}\sigma^{2}\right)\nonumber \\ &\mathcal{L}_{i}\left(\frac{k\eta\left(2^{\frac{S}{\omega \min(\tau_{x},{\tau_{y}})}}\!\!-\!\!1\right)d_{j,i}^{-\alpha}}{P_{t}}\right),
\end{align}
where $\eta=\beta(\beta!)^{-1/\beta}$, and 
\begin{align}
\label{laplace}
\mathcal{L}_{i}(n)&=\exp\left(-2\pi\lambda\int_{0}^{\infty}\left(1-\frac{1}{1+nP_{t}r^{-\alpha}}\right)rdr\right).
\end{align}
\end{lemma} 
\begin{proof}[Proof:\nopunct]
	The basics of the proof follow from \cite{zeng2018joint} and are omitted here due to space constraints. However, the proof of the Laplace transform in (\ref{laplace}) is different and provided as follows:
	\begin{align}
	\mathcal{L}_{i}(n) =& \mathbb{E}_{\Phi} \left[ \exp\left(-n \sum_{c \in \Phi}P_{t} g_{c,i}(t)(d_{c,i}(t))^{-\alpha} \right)  \right] \nonumber \\ 
	=&\mathbb{E}_{\Phi}\left[\prod_{c \in \Phi}\mathbb{E}_{g_{c,i}}\left( \exp\left(-nP_{t} g_{c,i}(t)(d_{c,i}(t))^{-\alpha}\right)\right) \right] \nonumber\\
	\stackrel{(a)}{=}& \mathbb{E}_{\Phi}\left[\prod_{c \in \Phi} \frac{1}{1+nP_{t}d_{c,i}^{-\alpha}} \right] \nonumber \\ 
	\stackrel{(b)}{=}&\exp\left[-2\pi\lambda\int_{0}^{\infty}\left(1-\frac{1}{1+nP_{t}r^{-\alpha}}\right)rdr\right],
	\end{align}
	where $\Phi$ denotes the set of interfering UAVs around the swarm, 
	(a) follows from the assumption of Rayleigh fading channel and the channel gain $g_{c,i}(t)\sim\exp(1)$, and the probability generating function (PGFL) of 2D-PPP \cite{haenggi2012stochastic} is used to prove the changes in (b) .  
\end{proof}
\begin{table}[!t]
	\large
	\begin{center}
		\caption{\small Simulation parameters.}
		\label{table_example}
		\resizebox{8cm}{!}{
			\begin{tabular}{|c|c|c|}
				\hline
				\textbf{Parameter} & \textbf{Meaning} & \textbf{Value} \\  \hline
				$a_{2}$, $b_{2}$, $\hat{a}_{2}$, $\hat{b}_{2}$ & Associated gains of UAV 2& $1,1,1.5,1.5$ \\ \hline
				$a_{3}$, $b_{3}$, $\hat{a}_{3}$, $\hat{b}_{3}$ & Associated gains of UAV 2& $1,1,1.5,1.5$ \\ \hline				
				$k$ & Coefficient of nondecreasing function & $1.01$ \\ \hline
				$m$  & Nakagami parameter & $3$\\ \hline
				$\alpha$  & Path loss exponent & $3$ \\ \hline
				$\sigma^{2}$ & Noise variance  & $-174$~dBm/Hz  \\ \hline
				$S$ & Packet size & $3200$~bits \\ \hline 
				$\omega$ & Bandwidth & $20$~MHz \\ \hline 
		\end{tabular}}
	\vspace{-0.6cm}
	\end{center}
\end{table}
\begin{figure}[!t]
	\centering
	\subfloat[Spacing errors over time.]{\includegraphics[width=0.40\textwidth]{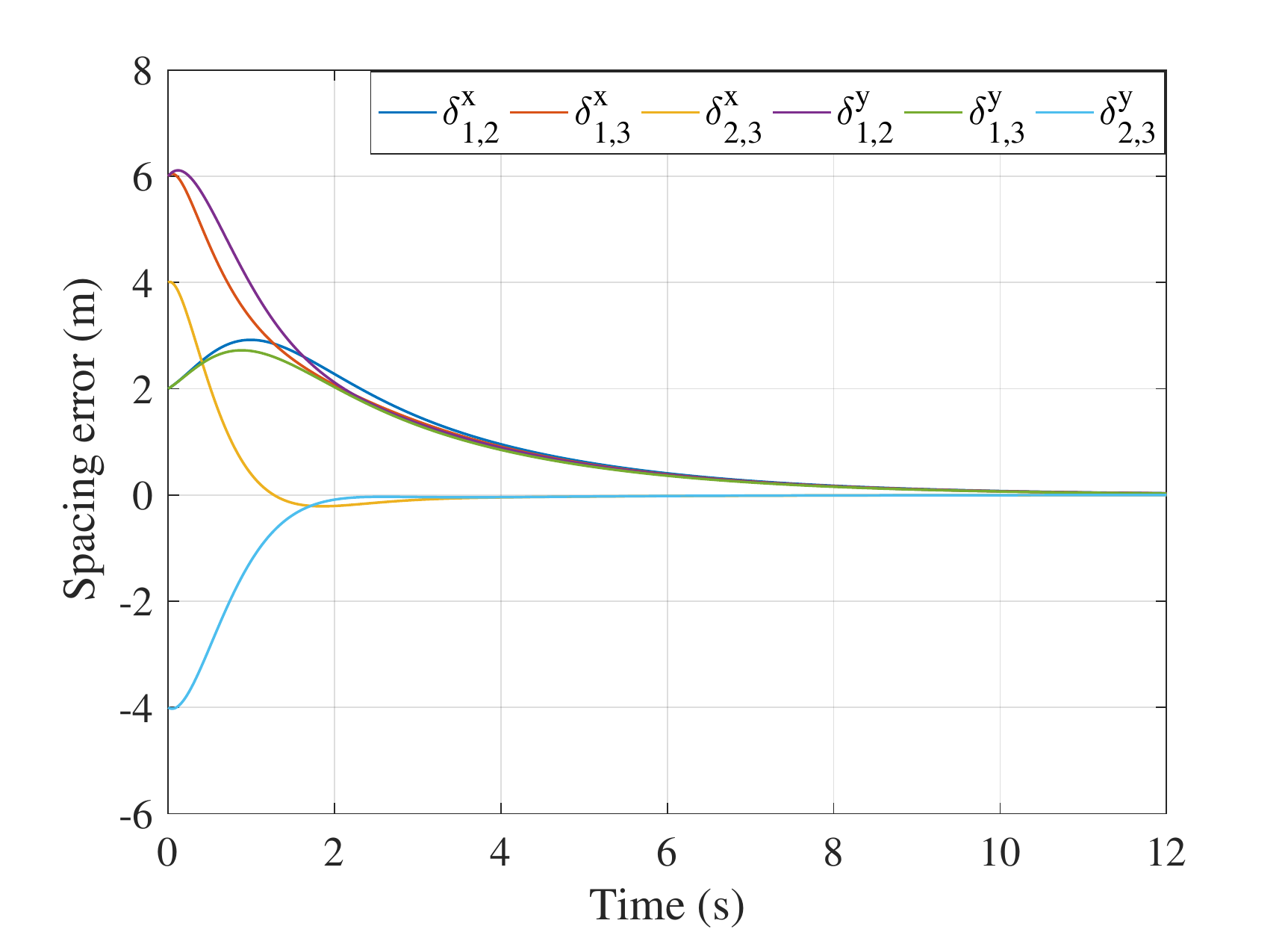}\label{fig:f1}}
	\vspace{-0.3cm}
	\hfill
	\subfloat[Velocity errors over time.]{\includegraphics[width=0.40\textwidth]{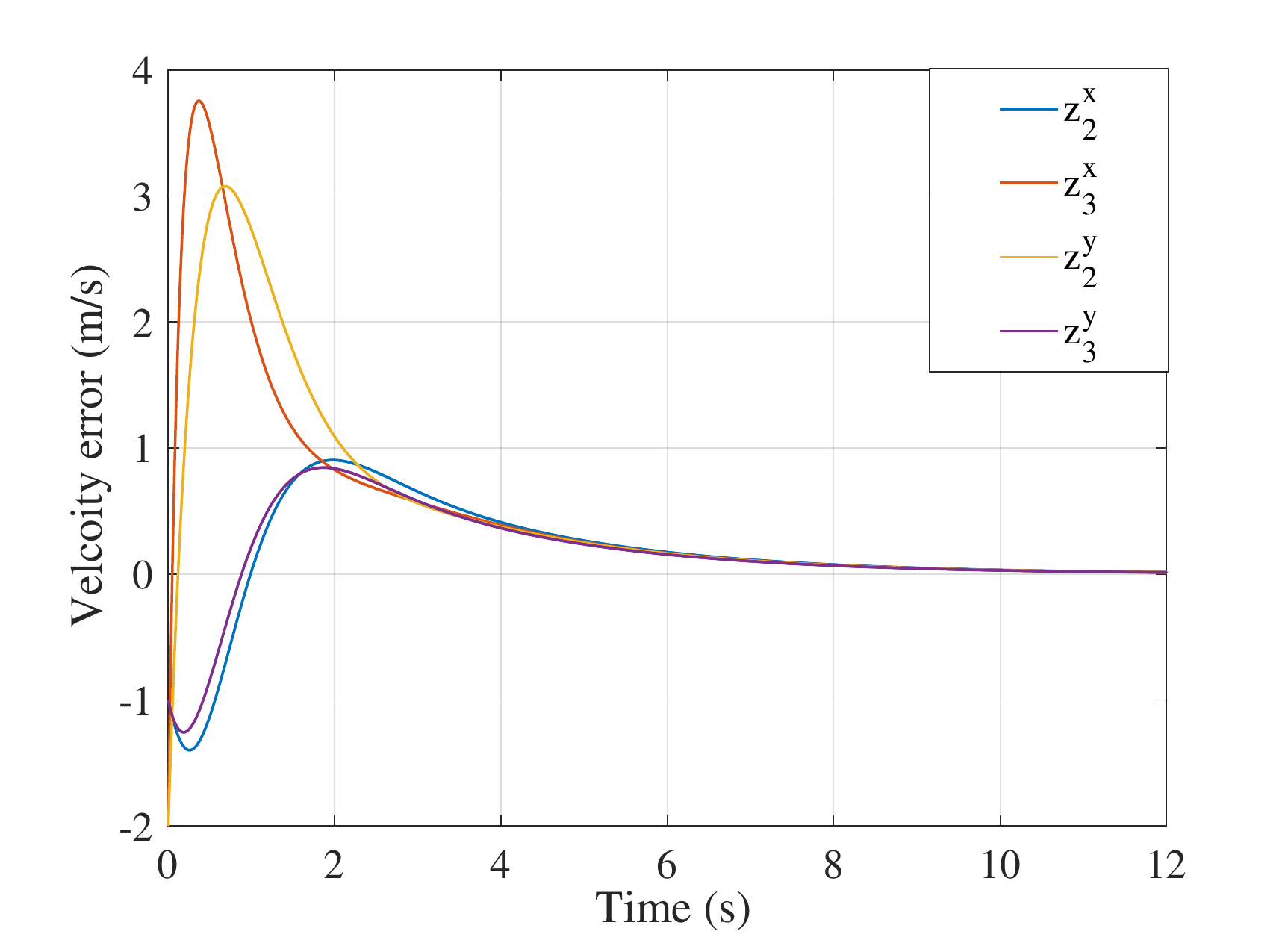}\label{fig:f2}}
	\vspace{-0.1cm}
	\caption{Stability validation for a swarm of three UAVs when using the maximum allowable transmission delay derived in Corollary \ref{theorem1}.}
	\label{figstability}
	\vspace{-0.5cm}
\end{figure}
\section{Simulation Results and Analysis}
For our simulations, we first validate the results derived in Corollary \ref{theorem1}.
Based on Lemma \ref{theorem2}, we then study the impact of interference on the reliability performance of the wireless network and finally obtain the design guideline of formulating a stable triangle formation for a swarm of three UAVs. 
All simulation parameters are summarized in Table \ref{table_example}. 
Without loss of generality, we assume that the two following UAVs have the same control gains, given in Table \ref{table_example}. 
Using the parameter settings in Table \ref{table_example} for Corollary \ref{theorem1}, we can find that the maximum allowable transmission delay to avoid the instability of the control system is $18.2$~ms.

We first corroborate the analytical result in Corollary \ref{theorem1} on the stability of the control system under the derived transmission delay threshold. 
In particular, we model the uncertainty of the wireless channel pertaining to the wireless communication links in the swarm of UAVs as a time-varying delay in the range $(0,18.2$~ms). 
The following UAVs are initially assigned with different velocities from the target velocity and random locations. 
Here, the leading UAV flies with speed components $\bar{v}_{x}=5$~m/s and $\bar{v}_{y}=5$~m/s, and the target spacing between UAVs are $\bar{x}_{1,2}=3$~m, $\bar{x}_{1,3}=4$~m, $\bar{x}_{2,3}=1$~m, $\bar{y}_{1,2}=4$~m, $\bar{y}_{1,3}=3$~m, and $\bar{y}_{2,3}=-1$~m.
Fig. \ref{figstability}\subref{fig:f1} shows the time evolution of the spacing errors. 
We can observe that the spacing errors at x-axis and y-axis for both following UAVs will eventually converge to $0$.
Also, in Fig. \ref{figstability}\subref{fig:f2}, we can observe that the velocity errors will converge to $0$ when time passes by.  
Thus, by choosing the maximum delay derived in Corollary \ref{theorem1}, we can ensure that the stability of the swarm of UAVs is guaranteed and the target formation can be formulated. 
\begin{figure}[!t]
	\centering
	\includegraphics[width=3in,height=2.5in]{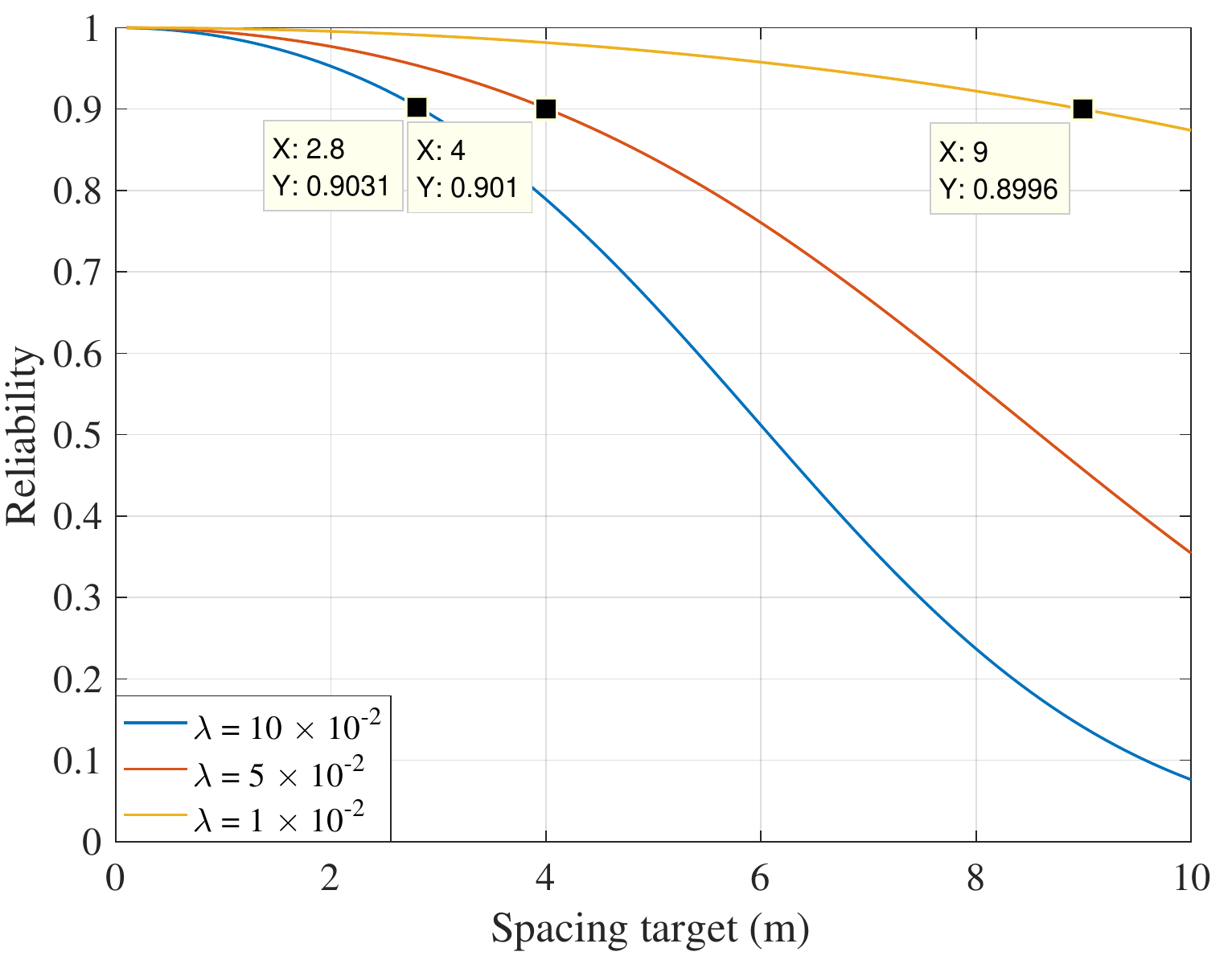}
	\vspace{-0.3cm}
	\DeclareGraphicsExtensions.
	\caption{Reliability performance of systems with different densities of interfering UAV when spacing target increases.}
	\label{Reliability}
	\vspace{-0.6cm}
\end{figure}

Fig. \ref{Reliability} shows the reliability performance of the wireless network with different densities of interfering UAVs when the spacing target increases.  
As observed from Fig. \ref{Reliability}, when the spacing target between two following UAVs increases, the reliability of the wireless network will decrease. 
For example, for a system with density of interfering UAVs  $\lambda=5\times 10^{-2}$~UAV/m$^2$, the reliability is around $35$\% when the spacing target is $10$~m. 
However, when the spacing target is $4$~m, the reliability is $90.1$\%.
Moreover, we can obtain a design guideline on how to guarantee a stable UAV formation from the results shown in Fig. \ref{Reliability}.
In particular, to guarantee that the reliability of the wireless system exceeds a threshold, we need to properly choose the target spacing  between two following UAVs in the swarm. 
As shown in Fig. \ref{Reliability}, for a system with $\lambda=10\times 10^{-2}$~UAV/m$^2$, the target spacing should be chosen a smaller value than $2.8$~m so that the reliability performance can exceed $90$\%. 
Also, for a system with $\lambda=5\times 10^{-2}$~UAV/m$^2$, the corresponding target spacing value should be smaller than $4$~m to reach a reliability of $90$\%.
Moreover, when the density is chosen as $\lambda=1\times 10^{-2}$~UAV/m$^2$, the spacing target should be smaller than $9$~m.
This is due to the fact that the strength of the receiving signal will decrease and the delay of the communication link will increase when the spacing increases. 

\section{Conclusion}
In this paper, we have proposed a novel approach to jointly design the control and communication system of a cellular-connected swarm of UAVs.
Based on the integrated communication and control strategy, we have performed a control system stability analysis and derived the delay threshold which can prevent the instability of the swarm of UAVs. 
We have used stochastic geometry to derive the mathematical expression for the reliability of the wireless system, defined as the probability of meeting the control system's delay requirements. 
Simulation results have shown that leveraging the synergies between control and wireless systems can result in a stable operation of a swarm of UAVs. 

\def\baselinestretch{0.97}
\bibliographystyle{IEEEtran}

\IEEEpeerreviewmaketitle

\end{document}